\documentclass[12pt,reqno]{article}

\usepackage[usenames]{color}
\usepackage{amssymb}
\usepackage{amsmath}
\usepackage{amsthm}
\usepackage{amsfonts}
\usepackage{amscd}
\usepackage{graphicx}
\usepackage{array}
\usepackage{xcolor}

\usepackage[colorlinks=true,
linkcolor=webgreen,
filecolor=webbrown,
citecolor=webgreen]{hyperref}

\definecolor{webgreen}{rgb}{0,.5,0}
\definecolor{webbrown}{rgb}{.6,0,0}

\usepackage{color}
\usepackage{fullpage}
\usepackage{float}

\usepackage{graphics}
\usepackage{latexsym}
\usepackage{epsf}
\usepackage{breakurl}

\usepackage[shortlabels]{enumitem}

\newcommand{\ie}{i.e.,}
\newcommand{\eg}{e.g.,}
\newcommand{\etal}{et al.}

\newcommand{\seqnum}[1]{\href{https://oeis.org/#1}{\underline{#1}}}
\def\modd#1 #2{#1\ \mbox{\rm (mod}\ #2\mbox{\rm )}}

\def\modd#1 #2{#1\ \mbox{\rm (mod}\ #2\mbox{\rm )}}

\newcolumntype{L}[1]{>{\raggedright\let\newline\\\arraybackslash\hspace{0pt}}m{#1}}

\begin{document}

\theoremstyle{plain}
\newtheorem{theorem}{Theorem}
\newtheorem{corollary}[theorem]{Corollary}
\newtheorem{lemma}[theorem]{Lemma}
\newtheorem{proposition}[theorem]{Proposition}

\theoremstyle{definition}
\newtheorem{definition}[theorem]{Definition}
\newtheorem{example}[theorem]{Example}
\newtheorem{conjecture}[theorem]{Conjecture}

\theoremstyle{remark}
\newtheorem{remark}[theorem]{Remark}

\title{New Bounds on Antipowers in Words}

\author{Lukas Fleischer, Samin Riasat, and Jeffrey Shallit\footnote{Research supported in part by NSERC Grant 2018-04118.} \\
School of Computer Science \\
University of Waterloo\\
Waterloo, ON  N2L 3G1 \\
Canada\\
{\tt \{lukas.fleischer,samin.riasat,shallit\}@uwaterloo.ca}}

\maketitle

\begin{abstract}
Fici et al.~defined a word to be a $k$-power if it is the concatenation of $k$ consecutive identical blocks, and an $r$-antipower if it is the concatenation of $r$ pairwise distinct blocks of the same size.  They defined $N(k,r)$ as the smallest $\ell$ such that every binary word of length $\ell$ contains either a $k$-power or an $r$-antipower.  In this note we obtain some new upper and lower bounds on $N(k,r)$.  We also consider avoiding $3$-antipowers and $4$-antipowers over larger alphabets, and obtain a lower bound for $N(k,5)$ in the binary case.
\end{abstract}

\section{Introduction}

Regularities and repetitions have been studied extensively in the field of combinatorics on words.
One of the early results in the area is Thue's observation that while every sufficiently long binary word contains a square, in contrast there are arbitrarily long words over a ternary alphabet avoiding squares ~\cite{Berstel:1995,Thue:1906, Thue:1912}. In this context, \emph{avoidance} of a certain set of words means that none of the words of this set appears as a factor.
Thue's results show that avoidance of powers depends on the alphabet size. In this note, we focus solely on binary words.
The study of avoidance of patterns has been extended to \emph{$k$-powers}, \ie~nonempty words of the form $u^k = \overbrace{uu\cdots u}^k$, and other variants of the problem; see \eg~\cite{Dekking76,EntringerJS74,Fraenkel&Simpson:1994}.

In~\cite{FiciRSZ16}, Fici \etal~introduced the notion of antipowers.
Whereas a power is a sequence of adjacent blocks that are all the same, an antipower is a sequence of consecutive adjacent blocks of the same length that are pairwise different.
Formally, a \emph{$k$-antipower} is a word of the form $u_1 \cdots u_k$ such that $|u_1| = \cdots = |u_k|$ and the factors $u_1, \dots, u_k$ are pairwise distinct, \ie~$|\{u_1, \dots, u_k\}| = k$.

Fici \etal~also suggested investigating the simultaneous avoidance of powers and antipowers, which is the main topic of this work.
They defined $N(k,r)$ to be the shortest length~$\ell$ such that every binary word of length $\ell$ contains either a $k$-power or an $r$-antipower as a factor.
It is known that $N(k, r)$ is bounded polynomially in $k$ and $r$~\cite{FiciRSZ16}.
The available numerical evidence~(see \cite{Riasat19}) suggests that for each
fixed $r \geq 2$ we have 
$$(2r-4) k \le N(k,r) \le (2r-4) k + f(r)$$
for some function $f \colon \mathbb{N} \to \mathbb{N}$.
We first prove the lower bound $N(k,r) \ge (2r-4)k$ for $k \ge 4$.
Then we show that this bound is tight for $r = 3$.

We then address avoiding $3$-antipowers and $4$-antipowers over arbitrary alphabets.   Finally, we close with a discussion of avoiding $5$-antipowers over the binary alphabet.

This paper is based, in part, on the master's thesis of the second author \cite{Riasat19}.

Throughout the paper, we use $\varepsilon$ to denote the empty word. Moreover, for a word $w$, we use $w^*$ to denote the set of words $\{ \varepsilon, w, w^2, w^3, \dots \}$.   By $x^\omega$ we mean the (one-sided) infinite word $xxx\cdots$.

\section{A lower bound on $N(k,r)$}

To prove the desired lower bound, we give an explicit family of binary words $(x_{k,r})$ and use combinatorial arguments to show that $x_{k,r}$ avoids $k$-powers and $r$-antipowers. In what follows, we use the standard notation for regular expressions, as discussed (for example) in~\cite{Hopcroft&Ullman:1979}.

\begin{theorem}
Let $r \geq 3$ and $k \geq \max\{r-1,4\}$.  Define
$x_{k,r} = ((01)^{k-1} 00)^{r-3} (01)^{k-1} 0$.  Then
$|x_{k,r}| = 2k(r-2)-1$, but $x_{k,r}$ has
no $k$-powers, nor $r$-antipowers.
\label{thm:lower}
\end{theorem}

\begin{proof}
First, we argue that $x_{k,r}$ contains no $k$-power $y^k$.
Suppose it does.
Clearly, $x_{k,r}$ does not contain $0000$ or $1111$, so $|y| \ge 2$. However, every factor of length $2k$ of $x_{k,r}$ contains the word $00$, and so $y^k$ contains $00$.  
Thus, $y^2$ must also contain the factor $00$ and $y^3$ must contain $00$ at least twice. If two occurrences of $00$ appear at a distance of less than $2k$ in $y^3$, then $y^k$ cannot be a factor of $x_{k,r}$.
Therefore, $y$ must have length at least $2k$. But then $y^k$ has length at least $2k^2 \ge 2k(r-1) = 2k(r-2) + 2k > |x_{k,r}|$, contradicting our assumption.

Next, we argue that $x_{k,r}$ contains no $r$-antipower.
First, observe that in any sequence of non-overlapping blocks in $x_{k,r}$, at most $r-3$ blocks can contain $00$ as a factor. All remaining blocks belong to the set $\{\varepsilon,1\}(01)^*\{\varepsilon,0\}$. Since the set $\{\varepsilon,1\}(01)^*\{\varepsilon,0\}$ only contains two words of each length, this implies that any sequence of $r$ consecutive blocks of the same length contains at least two identical blocks. We conclude that $x_{k,r}$ contains no $r$-antipower.
\end{proof}

This theorem immediately yields a lower bound that matches our conjectured upper bound up to an additive term that only depends on $r$.

\begin{corollary}
$N(k,r) \geq (2r-4)k$
for all $k \geq \max\{r-1,4\}$.
\label{crl:lower}
\end{corollary}

\begin{remark} Since 
$N(2,4) = 4$ and
$N(3,4) = 19$, Corollary~\ref{crl:lower} gives
$N(k,4) \geq 4k$ for
$k \geq 3$.  Computations show that $N(k,4) = 4k$ for $11 \leq k \leq 30$.    We conjecture
that this equality holds for all $k \geq 11$.  

\end{remark}

\section{Binary words avoiding $3$-antipowers}

For a language $L \subseteq \{0,1\}^*$, we let $\mathcal{C}_n(L)$ denote the (transitive) closure of $L \cap \{0,1\}^n$ under bitwise complementation and reversal.
It consists of all length-$n$ words from $L$, all bitwise complements of length-$n$ words from $L$, all reversed length-$n$ words from $L$ and all bitwise complements of reversed length-$n$ words from $L$.
It is easy to see that if a binary word contains a $k$-power, then both its bitwise complement and its reversal also contain a $k$-power. The same argument applies to $r$-antipowers.
Together with the fact that bitwise complementation and reversal are involutions, this implies that $\mathcal{C}_n(L)$ avoids $k$-powers (resp., $r$-antipowers) if and only if $L \cap \{0, 1\}^n$ avoids $k$-powers (resp., $r$-antipowers).

The definition of $\mathcal{C}_n$ allows us to give a simple description of all binary words avoiding $3$-antipowers.
Before giving the general result, we only characterize words avoiding $3$-antipowers whose lengths are multiples of $3$.

\begin{theorem}
  Let $n \ge 18$ be divisible by $3$.
  Let $A_n$ be the set of binary words of length~$n$ avoiding $3$-antipowers.
  Then $A_n = \mathcal{C}_n\left(0^* \cup (01)^* \cup (01)^* 0 \cup 0^*10^* \cup 0^*011 \cup 0^*101\right)$.
  \label{thm:three-antipowers}
\end{theorem}
\begin{proof}
  Note that it is easy to verify the claim for $n = 18$ by enumerating all words of this length. For larger values, we prove the claim by induction.

  Let $n > 18$ with $n \equiv \modd{0} {3}$.
  By induction, all factors of length $n - 3$ avoid $3$-antipowers. Moreover, it is easy to see that when splitting any word from $A_n$ into three blocks of equal length, at least two of these blocks coincide. Therefore, all words in $A_n$ avoid $3$-antipowers.
  It remains to show that if $w \in \{0, 1\}^n$ avoids $3$-antipowers, then $w \in A_n$. Since avoidance of $3$-antipowers is invariant under bitwise complementation, we may assume that $|w|_0 \ge |w|_1$.
  Moreover, if $|w|_0 = |w|_1$, we may assume that $w$ starts with a $0$.
  
  We factorize $w = xyz$ with $|x| = |z| = 3$. By induction, $xy$ and $yz$ belong to $A_{n-3}$. Since $w$ contains at least as many zeroes as ones, and starts with a $0$ if the number of zeroes equals the number of ones, this implies that $y \in 0^* \cup (10)^* \cup (10)^* 1 \cup 0^*10^*$.
  
  If $y \in 0^*10^*$, then $x = z = 000$, otherwise either $xy \not\in A_{n-3}$ or $yz \not\in A_{n-3}$.
  This implies $w \in 0^*10^*$, thus $w \in A_n$.
  Similarly, if $y \in (10)^* \cup (10)^* 1$, then $w \in (01)^* \cup (01)^* 1$ and $w \in A_n$.
  
  The remaining case is $y \in 0^*$. Since at least two of the factors $w_1, w_2, w_3$ in the unique factorization $w = w_1 w_2 w_3$ with $|w_1| = |w_2| = |w_3|$ must coincide, either $x = 000$ or $z = 000$. Avoidance of $3$-antipowers is invariant under reversal, so we may assume $x = 000$. Since $y \in 0^*$ implies that $yz$ belongs to $0^* \cup 0^*10^* \cup 0^*011 \cup 0^*101$ and this set is closed under prepending zeroes, we obtain that $w$ has the desired form. This concludes the proof.
\end{proof}

We now extend this characterization to words whose lengths are not divisible by $3$.

\begin{corollary}
  Let $A_n$ be the set of binary words of length~$n$ avoiding $3$-antipowers. Let $k \ge 6$. Then
  \begin{itemize}
      \item $A_{3k} = \mathcal{C}_{3k}\left(0^* \cup (01)^* \cup (01)^* 0 \cup 0^*10^* \cup 0^*011 \cup 0^*101\right)$,
      \item $A_{3k+1} = \mathcal{C}_{3k+1}\left(0^* \cup (01)^* \cup (01)^* 0 \cup 0^*10^* \cup 0^*011 \cup 0^*101 \cup 10^*1\right)$ and
      \item $A_{3k+2} = \mathcal{C}_{3k+2}\left(0^* \cup (01)^* \cup (01)^* 0 \cup 0^*10^* \cup 0^*011 \cup 0^*101 \cup 10^*1 \cup 10^*10 \cup 1 0^*11 \right)$.
  \end{itemize}
  In particular, for $k \geq 6$ there are exactly $6k + 12$ (resp., $6k + 16$, $6k + 26$) binary words of length $3k$ (resp., $3k+1$, $3k+2$) avoiding $3$-antipowers.
  \label{crl:three-antipowers}
\end{corollary}
\begin{proof}
  For words of length $3k + 1$, it suffices to investigate their two factors of length~$3k$ and apply the previous theorem.
  Similarly, for words of length $3k + 2$, we investigate all three factors of length~$3k$.
  
  We now fix some $k \ge 6$ and count the number of words of length $3k$.
  Words from $0^*$ and $(01)^* 0$ coincide with their reverses and words from $(01)^*$ coincide with the bitwise complements of their reverses. Thus, $\mathcal{C}_{3k}\left(0^* \cup (01)^* \cup (01)^* 0\right)$ contains exactly four words.
  The set $\mathcal{C}_{3k}(0^*10^*)$ contains $6k$ words: all words of the form $0^i 1 0^j$ or $1^i 0 1^j$ with $0 \le i < 3k$ and $i + j = 3k - 1$.
  The sets $0^*011$ and $0^*101$ contain exactly one word of length $3k$; and the reverse, the bitwise complement and the bitwise complement of the reverse of each of these words are distinct from one another. Therefore, $\mathcal{C}_{3k}\left(0^*011 \cup 0^*101\right)$ contains eight words.
  Together, this shows that $|A_{3k}| = 4 + 6k + 8 = 6k + 12$.

  Since every word from $10^*1$ coincides with its reverse but not with its bitweise complement, we obtain $|A_{3k+1}| = 4 + (6k + 2) + 8 + 2 = 6k + 16$.
  Moreover, every word from $10^*10 \cup 1 0^*11$ does not coincide with its reverse, its bitwise complement or the bitwise complement of its reverse, so $|A_{3k+2}| = 4 + (6k + 4) + 8 + 2 + 8 = 6k + 26$.
\end{proof}

To extend Corollary~\ref{crl:three-antipowers} to infinite words, it suffices to determine all words whose factors belong to the set $A_n$ defined in Theorem~\ref{thm:three-antipowers}.

\begin{corollary}
The only infinite binary words avoiding $3$-antipowers are of the form
$0^\omega$, $(01)^\omega$,  $11 0^\omega$, $101 0^\omega$, $0^i 1 0^\omega$ for $i \geq 0$,
and their binary complements.
\end{corollary}

Using Corollary~\ref{crl:three-antipowers}, we can prove a tight upper bound for $N(k, 3)$.

\begin{theorem}
  $N(k, 3) = 2k$ for all $k \ge 9$.
\end{theorem}
\begin{proof}
  The lower bound follows from Corollary~\ref{crl:lower}. Thus, it suffices to show that every binary word of length $2k$
  for $k \geq 9$, contains either a $k$-power or a $3$-antipower.

  Let $w \in \{0, 1\}^{2k}$ be a word that avoids $3$-antipowers.
  By closure of the sets of $k$-powers and $3$-antipowers under bitwise complementation, we may assume that $|w|_0 \ge |w|_1$.
  If $|w|_1 \le 1$, then $0^k$ is a factor of $w$.
  If $|w|_1 > 3$, then $w \in \{(01)^k, (10)^k\}$ by Corollary~\ref{crl:three-antipowers}.
  Moreover, if $|w|_1 \in \{2, 3\}$, then $w = xyz$ for some $x, y, z \in \{0, 1\}^*$ with $|x|, |z| \le 3$ and $y = 0^{2k-|xz|}$.
  In any case, $w$ contains a $k$-power.
\end{proof}

\begin{remark}
Our results above only hold for words of length at least $18$.
With a slightly more technical proof or by computing all words of length less than $18$ avoiding $3$-antipowers, it can be shown that the characterization given in Corollary~\ref{crl:three-antipowers} can be extended to all words of length at least~$12$.
Moreover, it is easy to see that all words of length $\le 5$ avoid $3$-antipowers.
For the remaining lengths, the following table lists all words avoiding $3$-antipowers up to bitwise complementation and reversal.

\begin{center}
\begin{tabular}{|c|L{14.95cm}|}
$n$ & Words of length $n$ avoiding $3$-antipowers, up to complementation and reversal \\
\hline
6 & \footnotesize $000000$, $000100$, $001100$, $000010$, $001010$, $010110$, $000001$, $010001$, $011001$, $000101$, $100101$, $010101$, $000011$ \\
7 & \footnotesize $0000000$, $0001000$, $0000100$, $0010100$, $0000010$, $0100010$, $0001010$, $0101010$, $0000001$, $1000001$, $0010001$, $0101001$, $0011001$, $0000101$, $1000101$, $0010101$, $0000011$ \\
8 & \footnotesize $00000000$, $00001000$, $00000100$, $00010100$, $00000010$, $00100010$, $00001010$, $00101010$, $01100110$, $01010110$, $00000001$, $10000001$, $01000001$, $00010001$, $01010001$, $00101001$, $10011001$, $00000101$, $01000101$, $00010101$, $10010101$, $01010101$, $00000011$, $10000011$, $00110011$, $00101011$ \\
9 & \footnotesize $000000000$, $000010000$, $000001000$, $000101000$, $000000100$, $000000010$, $010000010$, $010001010$, $010101010$, $000000001$, $001000001$, $001010001$, $000000101$, $101000101$, $000000011$ \\
10 & \footnotesize $0000000000$, $0000010000$, $0000001000$, $0000000100$, $0000000010$, $0010000010$, $0010100010$, $0000000001$, $1000000001$, $0001000001$, $0001010001$, $0000000101$, $0101000101$, $0101010101$, $0000000011$ \\
11 & \footnotesize $00000000000$, $00000100000$, $00000010000$, $00000001000$, $00000000100$, $00100000100$, $00000000010$, $00010000010$, $00010100010$, $01010001010$, $01010101010$, $00000000001$, $10000000001$, $01000000001$, $00001000001$, $10001010001$, $00000000101$, $00101000101$, $00000000011$, $10000000011$
\end{tabular}
\end{center}
\label{remark9}
\end{remark}

\section{Avoiding $3$-antipowers over larger alphabets}

Call a word $w \in \Sigma_k^*$ {\it orderly\/} if
$i < j$ implies that the first occurrence of $i$ in $w$
precedes the first occurrence of $j$ in $w$.  Clearly,
every word has the property that, by renaming the letters in
a 1--1 fashion, one can obtain an orderly word.

\begin{theorem}
Let $|w| = n \geq 9$, and suppose $w$ avoids $3$-antipowers and
contains exactly three distinct letters.   Then, up to renaming of
the letters, we have the following:
\begin{enumerate}[(a)]
\item If $|w| \equiv \modd{0} {3}$, then there is no such $w$;

\item If $|w| \equiv \modd{1} {3}$, then $w = 0 1^{n-2} 2$;

\item If $|w| \equiv \modd{2} {3}$, then $w$ is one of the
following:  $001^{n-3} 2$, $010^{n-3}2$,
$0 1^{n-2} 2$, $0 1^{n-3} 21$, or $01^{n-3}22$.
\end{enumerate}
\label{thm10}
\end{theorem}

\begin{proof}
  We can verify the claim for $|w| < 18$ by exhaustive enumeration and consider only the case $n := |w| \ge 18$.
  Let $a, b, c$ denote the three letters occurring in $w$, in descending order of number of occurrences in $w$.
  Define a coding $\tau$ by $\tau(a) = 0$ and $\tau(b) = \tau(c) = 1$. Note that $\tau(w)$ also avoids $3$-antipowers, and contains at least six zeros and at least two ones.

  If $\tau(w)$ contains more than three occurrences of $1$, then $\tau(w) \in \mathcal{C}_{n}((01)^* \cup (01)^*0)$ by Corollary~\ref{crl:three-antipowers}. But then applying one of the two possible codings $\tau'$ with $\tau'(a) = 0$ and $\{\tau'(b), \tau'(c)\} = \{0, 1\}$ to $w$ yields a word $\tau'(w)$ that is not listed in Corollary~\ref{crl:three-antipowers}. This contradicts the fact that $w$ avoids $3$-antipowers. Therefore ${|\tau(w)|}_1 \in \{2, 3\}$.
  
  If  ${|\tau(w)|}_1 = 3$, then ${|w|}_b = 2$ and ${|w|}_c = 1$.
  Corollary~\ref{crl:three-antipowers} yields $n \equiv \modd{2} {3}$ and $\tau(w) \in 10^*11 \cup 110^*1$. By symmetry, we may assume that $\tau(w) \in 10^*11$. If the last two letters of $w$ were different, the suffix of length $3$ of $\tau(w)$ would be a $3$-antipower. Therefore, up to renaming of letters and reversal, we have $w = 01^{n-3}22$.
  
  If  ${|\tau(w)|}_1 = 2$, then ${|w|}_b = {|w|}_c = 1$ and Corollary~\ref{crl:three-antipowers} yields $\tau(w) \in \mathcal{C}_{n}(0^*011 \cup 0^*101 \cup 10^*1 \cup 10^*10)$. Since the preimages of both occurrences of $1$ must be distinct, the relation $\tau(w) \in 0^*011 \, \cup \, 0^*101$ yields $3$-antipowers in the suffix of $w$ of length $3$ or length $6$, respectively. Thus either
  \begin{enumerate}
      \item $w = 01^{n-2}2$ up to renaming of letters and reversal and $(n \bmod 3) \in \{1, 2\}$, or
      \item $w = 01^{n-2}21$ up to renaming of letters and reversal and $n \equiv \modd{2} {3}$.\qedhere
  \end{enumerate}
\end{proof}

\begin{corollary}
Every $3$-antipower-free infinite word contains at most two distinct letters.
\end{corollary}

\begin{theorem}
Every word containing at least four distinct letters has a $3$-antipower.
\end{theorem}

\begin{proof}
Suppose there is a word with four distinct letters, say
$\{ 0,1,2,3 \}$, but having no
$3$-antipower. Choose a shortest such word, and call it
$w$.     Without loss of generality, we may assume that $w$ is orderly
and contains
exactly four distinct letters, for if $w$ contains more than four
distinct letters, say numbers larger than $3$,
then we could apply a coding $\tau$ to $w$ that maps
letters $\leq 3$ to themselves, but all larger letters to $3$.
If $w$ has no $3$-antipower, then neither does $\tau(w)$.  

By exhaustive enumeration we can check that $|w| \geq 9$.
Write $w = w'a$ for $a$ a single letter.  Then from the minimality
of $w$, it must be that $w \in \{0,1,2\}^*$ and $a = 3$.
We can now appeal to Theorem~\ref{thm10}.  In each case
except one, the resulting word ends in a $3$-antipower of length $3$.
In the one remaining case we have $w' = 01^{n-4}22$ for $n \equiv 
\modd{2} {3}$.   Then $w'a$ is itself a $3$-antipower.
\end{proof}

\begin{corollary}
Let
$A_3(k,n)$ be the number of length-$n$
words avoiding $3$-antipowers over an alphabet of size $k$.
Then for $n \geq 12$ we have
$$ A_3(k,n) =
\begin{cases}
k( (k-1)n + 5k-4), & \text{if $n \equiv \modd{0} {3}$}; \\
k( (k-1)n + k^2 + 3k - 3), & \text{if $n \equiv \modd{1} {3}$};\\
k( (k-1)n + 5k^2-5k+1), & \text{if $n \equiv
\modd{2} {3}$}.
\end{cases}
$$
\end{corollary}

\begin{proof}
  Suppose $w$ is a word avoiding $3$-antipowers of length at least $12$.
    Then $w$ either has one, two, or three distinct letters.
     We prove the claim for the case 
$n \equiv \modd{0} {3}$, and omit the details for the other cases. In the case of one letter, there are clearly $k$
such words.   In the case of two letters,
Corollary~\ref{crl:three-antipowers} and Remark~\ref{remark9} enumerate all the possibilities; there are 5 families of binary words, and we
can replace the $0$ and $1$ in these by any pair of 
distinct letters to get a word that avoids $3$-antipowers.  There are no examples in the case of three letters.    This gives a total of
$(2n+12-2)k(k-1)/2 + k = k( (k-1)n + 5k-4)$ possibilities.

The other cases can be handled similarly.
\end{proof}

\section{Avoiding $4$-antipowers and beyond}

We now investigate $k$-antipowers for $k \ge 4$ and show that they behave different from $3$-antipowers in many ways.
Our first observation is regarding the growth of the set of words containing $k$-antipowers.

\begin{remark}
Corollary~\ref{crl:three-antipowers} shows that there are only linearly many binary words of length $n$ avoiding $3$-antipowers.  In contrast, the construction given in \cite[Prop.~12]{FiciRSZ16} can be adapted to show that $u(n)$, the number of length-$n$ binary words avoiding $4$-antipowers, grows faster than any polynomial in $n$.  This can be proved as follows:  let ${\bf a} = (a_i)_{i \geq 1}$ be any  infinite  sequence of integers
satisfying
$a_1 = 1$ and $a_{i+1} \geq 4 a_i$.
For $j \geq 1$ define
the characteristic
sequence
$$ c_{\bf a}(j) = \begin{cases}
1, & \text{if $j = a_i$  for some $i$ }; \\
0, & \text{otherwise},
\end{cases}
$$
and ${\bf c}_{\bf a} = c_{\bf a} (1) c_{\bf a} (2) \cdots $.
Then it is not hard to show, using the ideas in \cite[Prop.~12]{FiciRSZ16}, that every length-$n$
prefix of every ${\bf c}_{\bf a}$ avoids $4$-antipowers.   Let $t(n)$ be the number of length-$n$ prefixes of words of this form; then $u(n) \geq t(n)$.  It is not hard to see that $t(n) = t(n-1) + t(\lfloor n/4 \rfloor)$ for $n \geq 4$.  Then, according to
\cite{deBruijn:1948}, we have $t(n) = n^{\Theta(\log n)}$.   

The sequence $t(n)$ is sequence
\seqnum{A330513} in the On-Line Encyclopedia of Integer Sequences~\cite{Sloane}, and the sequence
$u(n)$ is sequence \seqnum{A275061}.
The exact growth rate of $u(n)$ is apparently still not known.
\end{remark}

We previously established that for every $k \le 3$, there is a maximum alphabet size for which $k$-antipowers can be avoided. This no longer holds for $k \ge 4$.
Note that since every $k$-antipower contains a $(k-1)$-antipower, it suffices to prove this claim for $k = 4$.

\begin{theorem}
  For every $m \ge 1$, there exists a word that contains $m$ distinct letters and avoids $4$-antipowers.
\end{theorem}
\begin{proof}
  Define an (infinite) alphabet $\Sigma = \{a_0, a_1, \dots\}$ and a sequence of words $w_0, w_1, \dots$ by $w_0 := a_0$ and $w_{i+1} := w_i {a_0}^{3{|w_i|}-1} a_{i+1}$ for $i \ge 0$. For each $i \ge 0$, the word $w_i$ contains exactly $i+1$ distinct letters.
  
  We show by induction that for each $i \ge 0$, the infinite word $w_i a_0^{\omega}$ (and thus also $w_i$ itself) does not contain a $4$-antipower.   This claim clearly holds for $i = 0$. We now assume that the claim holds for an arbitrary $i$ and show that the claim holds for $i+1$ too.
  Let $uvxy$ be a factor of $w_{i+1} a_0^{\omega}$ with $\ell := |u| = |v| = |x| = |y|$. We show that $uvxy$ cannot be a $4$-antipower.
  If $uvxy$ does not contain the letter $a_{i+1}$, then this immediately follows from the induction hypothesis.
  If $a_{i+1}$ is contained in $u$ or in $v$, then $x = y = a_0^\ell$.
  If $a_{i+1}$ is contained in $x$, then $v = y = a_0^\ell$.
  If $a_{i+1}$ is contained in $y$, then $v = x = a_0^\ell$.
  In any case, at least two of the factors $u, v, x, y$ coincide, so $w_{i+1} a_0^{\omega}$ avoids $4$-antipowers.
\end{proof}

\section{Binary words avoiding $5$-antipowers}

Better lower bounds can be obtained for some $r$.
One such example is given below.

\begin{theorem}
Let $k \geq 14$.
Then $N(k,5) \geq 6k+4$.
If, additionally, $k \equiv \modd{5} {10}$, then $N(k,5) \geq 6k+5$.
\end{theorem}

\begin{proof}
The following table gives, for each $k \ge 14$, a word $w_k$ that contains no $k$ power and no $5$-antipower.
The construction depends on the equivalence class of $\modd{k} {10}$.
Therefore, the table contains separate rows for each $i \in \{0, \dots, 9\}$ where $k \equiv \modd{i} {10}$.

\begin{center}
\begin{tabular}{|c|c|}
$i$ & $w_k$ \\
\hline
0 & $00 (01)^{k-1} 00 (01)^{k-1} 00 (01)^{k-1} 000$ \\
1, 3, 7, 9 & $0^8 (01)^{k-3} 00 (01)^{k-1} 00 (01)^{k-1} 0 $ \\
2, 4, 8 & $(01)^{(k+2)/2} (10)^{k-1} 11 (10)^{k-1} 1 (10)^{(k+2)/2} $ \\
5 & $0000 (01)^{k-1} 00 (01)^{k-1} 00 (01)^{k-1} 00 $ \\
6 & $000 (01)^{k-1} 11 (01)^{k-1} 11 (01)^{k-1} 10$ 
\end{tabular}
\end{center}

To see that none of these words contains a $k$-power, we can use exactly the same argument as in the proof of Theorem~\ref{thm:lower}.

To prove avoidance of $5$-antipowers, we also resort to ideas from the proof of Theorem~\ref{thm:lower} but the arguments are slightly more technical.
Let us first investigate the case $k \equiv \modd{0} {10}$, \ie~words of the form
$$00 (01)^{k-1} 00 (01)^{k-1} 00 (01)^{k-1} 000.$$
Assume, contrary to what we want to prove, that such a word contains a factor of the form $u_1 u_2 u_3 u_4 u_5$ where all the $u_i$ have the same lengths and are pairwise distinct. In the following, we say that a binary word is \emph{alternating} if (and only if) it belongs to the set $\{\varepsilon,1\}(01)^*\{\varepsilon,0\}$.

Note that since the $(01)^{k-1}$ patterns in $w_k$ appear at even distances and since $|u_1u_2|$ and $|u_3u_4|$ have even lengths, only one of the three words $u_1, u_3, u_5$ can be alternating; otherwise, two of these words would coincide. Similarly, either $u_2$ or $u_4$ must not be alternating.
On the contrary, since $|w_k| = 6k + 3 \equiv \modd{3} {5}$, either the prefix $00$ or the suffix $00$ of $w_k$ does not overlap with the occurrence of $u_1 u_2 u_3 u_4 u_5$. Considering all possible factorizations of $w_k$ into five consecutive blocks of equal lengths, it is easy to see that at least two of the words $u_1, u_2, u_3, u_4, u_5$ must be alternating.
Together with the previous observation, this means that exactly one of the words $u_1, u_3, u_5$ and exactly of the words $u_2, u_4$ is alternating.
We now analyze the six possible cases.

First, suppose that $u_1$ and $u_2$ are alternating. This implies that the $u_i$ have length at most $k - 1$ because every factor of length $2k$ of $w_k$ contains $00$ as a prefix or $00$ as a suffix or $000$ as a factor. But then, it is easy to see that at either $u_3$ or $u_4$ are alternating as well, a contradiction.
A similar argument applies whenever two adjacent blocks are alternating.
The remaining cases are that $u_1, u_4$ are alternating or $u_2, u_5$ are alternating. Since they are symmetrical, it suffices to investigate the first case.
Considering the possible positions of $u_1$ and $u_4$ in $w_k$, it becomes obvious that then, either $u_2$ or $u_3$ has to be alternating as well, again contradicting the observation above.

This concludes the correctness proof of the construction.
Similar arguments can be used to prove avoidance of $5$-antipowers for the lower bound witnesses given for $i \in \{1, \dots, 9\}$.
\end{proof}

\section{Open problems}

The exact asymptotic behavior of $N(k, r)$ remains open.
To the best of our knowledge, the best general upper bound known to date is linear in $k$ and quartic in $r$ \cite{Riasat19}, a bound that actually holds for all alphabet sizes.
However, it is consistent with this result and our conjecture that the asymptotic behavior is of the form $(2r - 4)k + f(r)$ for some function $f(r) \in \mathcal{O}(r^4)$.
While still linear in $k$ and quartic in $r$, proving an upper bound of the form $ckr + f(r)$ for some constant $c \in \mathbb{N}$ and some function $f(r) \in \mathcal{O}(r^4)$ would be a big step towards proving our conjecture.

It would also be interesting to investigate the growth of $f$ more carefully.
Can our lower bound be improved to show that $f$ is unbounded? Can we prove a subquartic upper bound?
More numerical evidence might help establish a conjecture on the growth rate of $f$.

\bibliographystyle{new2}
\bibliography{antipowers}

\end{document}